\documentclass{amsart}
\usepackage{graphicx}
\usepackage{amscd}
\usepackage{amsmath}
\usepackage{amsfonts}
\usepackage{amssymb}
\usepackage{bbm}
\usepackage{setspace}
\usepackage{enumerate}         
\usepackage{fixme}
\usepackage{color}
\usepackage{url}
\usepackage{amsthm}
\usepackage{bm}
\usepackage{xy}
\usepackage{enumitem}
\usepackage{mathrsfs}
\usepackage[gen]{eurosym}
\usepackage{todonotes}

\theoremstyle{plain}
\newtheorem{theorem}{Theorem}[section]

\newtheorem{lemma}[theorem]{Lemma}

\newtheorem{proposition}[theorem]{Proposition}

\newtheorem{definition}[theorem]{Definition}

\theoremstyle{remark}

\numberwithin{equation}{section}

\newcommand{\ind}{1\!\kern-1pt \mathrm{I}}
\newcommand{\rsto}{]\!\kern-1.8pt ]}
\newcommand{\lsto}{[\!\kern-1.7pt [}

\numberwithin{equation}{section}

\newcommand{\Hess}{\operatorname{Hess}}

\renewcommand{\rho}{\varrho}

\begin{document}
\title[On portfolios generated by optimal transport]{On portfolios generated by optimal transport}

\begin{abstract}
First introduced by Fernholz in stochastic portfolio theory, functionally generated portfolio allows its investment performance to be attributed to directly observable and easily interpretable market quantities. In previous works \cite{PW15, PW16} we showed that Fernholz's multiplicatively generated portfolio has deep connections with optimal transport and the information geometry of exponentially concave functions. Recently, Karatzas and Ruf \cite{KR17} introduced a new additive portfolio generation whose relation with optimal transport was studied by Vervuurt \cite{V16}. We show that additively generated portfolio can be interpreted in terms of the well-known dually flat information geometry of Bregman divergence. Moreover, we characterize, in a sense to be made precise, all possible forms of functional portfolio constructions that contain additive and multiplicative generations as special cases. Each construction involves a divergence functional on the unit simplex measuring the market volatility captured, and admits a pathwise decomposition for the portfolio value. We illustrate with an empirical example.%
\end{abstract}

\subjclass[2000]{}

\author{Ting-Kam Leonard Wong}
\address{University of Southern California, Los Angeles, CA 90089, United States}
\date{\today}
\maketitle %

\section{Introduction}\label{sec:intro}
An important problem in portfolio management is performance attribution. For a general investment algorithm, it is often difficult to explain the PnL (profit and loss) of the portfolio transparently in terms of market movements. Mathematically, we may think of the portfolio value as an integral of the trading strategy with respect to the underlying price process (see for example the first chapter of \cite{F02} and \cite{KS98}). This is generally a complicated function of the path taken by the market and does not admit convenient simplifications. %

In \cite{F99} Fernholz introduced a portfolio construction and proved a pathwise decomposition formula for the value process, when the market portfolio is taken as the numeraire. Called functionally generated portfolios, they are now key tools of stochastic portfolio theory as explained in \cite{F02, FK09}. A particularly important consequence is that they allow us to formulate simple structural conditions on large equity markets under which it is possible to outperform the market portfolio. For examples of such relative arbitrages with respect to the market portfolio as well as the structural conditions required, we refer the reader to the papers \cite{F99, FK05, FKK05, BF08, VK15, P16, FKR16, KR17} and their references. Following \cite{KR17}, we say that Fernholz's portfolios are {\it multiplicatively} generated -- this terminology will become clear in Section \ref{sec:fgp}. %

While the aforementioned papers worked in continuous time where the stock prices are It\^{o} processes, in a series of papers \cite{PW13, PW15, PW16} we worked in a discrete time and model-independent framework, and discovered an elegant connection between multiplicatively generated portfolios and optimal transport with a logarithmic cost function. For general introductions to optimal transport see \cite{V03, V08}. Furthermore, in \cite{PW16} we showed that optimal rebalancing of these portfolios can be interpreted in terms of a new information geometry of exponentially concave function and $L$-divergence (see Definition \ref{def:L.divergence}). Mathematically, the work \cite{PW16} reveals a previously unexplored link between optimal transport and information geometry and identified important examples beyond the classical Wasserstein costs.  See \cite{P17} for an extension of this connection to cost functions given by the cumulant generating functions of exponential families. Information geometry originated in the study of statistical inference using ideas of differential geometry. After several decades of intense development, it has now numerous applications in statistics, information theory and machine learning. See \cite{A16} for an introduction to this beautiful area. In our context, the main idea is that the trading strategy, given by an exponentially concave function, induces a geometric structure on the unit simplex which represents the states of the stock market. The financial gains and losses of the portfolio can then be quantified using geometric concepts such as divergences (in the sense of Definition \ref{def:divergence} below) and angles. %

On the other hand, a new class of {\it additively} generated portfolios was introduced by Karatzas and Ruf in \cite{KR17}. For these portfolios it is also possible to derive a pathwise decomposition which is analogous to Fernholz's formula. Recently, it was observed in the thesis \cite{V16} that additively generated portfolios also correspond to an optimal transport problem. Here the cost function is the Euclidean inner product which is equivalent to the classical quadratic cost $c(x, y) = |x - y|^2$. Analogous to \cite{PW16}, we will see in this paper that these portfolios are connected to the well known dually flat geometry of Bregman divergence introduced by Amari and Nagaoka in \cite{NA82}. We emphasize that these geometric ideas becomes apparent only in discrete time. This is because in continuous time all higher order effects except the quadratic variation (which is essentially the Riemannian metric) vanish. %

Some natural questions are the following: What is the relationship between multiplicatively and additively generated portfolios? Are there other forms of functional portfolio generation that are connected to optimal transport and information geometry? In this paper we answer these questions. We consider a general framework of functional portfolio construction (see Definition \ref{def:generation}) and characterize all its possible forms. As will become clear, functional portfolio generation, even in our extended form, imposes strong conditions on the feasible trading strategies. We find that apart from multiplicative and additive generation there exists a new family of portfolio constructions that can be parameterized by two parameters $\alpha > 0$ and $C \geq 0$. We show that each trading strategy $\eta$ in this class is a long-short portfolio of a multiplicatively generated portfolio and the market portfolio. Moreover, we prove in Theorem \ref{thm:general.decomp} a discrete time pathwise decomposition of the value process $V_{\eta}(t)$ when the market portfolio is taken as the numeraire and $V_{\eta}(\cdot) > -C$:
\begin{equation*}
\frac{1}{\alpha} \log \frac{C + V_{\eta}(t)}{C + V_{\eta}(0)} = \varphi(\mu(t)) - \varphi(\mu(0)) + \sum_{s = 0}^{t - 1} D\left[\mu(s + 1) \mid \mu(s) \right].
\end{equation*}
Here $\varphi$ is the generating function which is $\alpha$-exponentially concave (i.e., $e^{\alpha \varphi}$ is concave), $D\left[ \cdot \mid \cdot\right] \geq 0$ is a divergence (in the sense of information geometry) on the unit simplex $\Delta_n$  (here $n$ is the number of stocks)  that we call the $L^{(\alpha)}$-divergence, and $\mu(t) \in \Delta_n$ is the market weight at time $t$. In this framework, multiplicative generation corresponds to the case $(\alpha, C) = (1, 0)$, and additive generation is the limit when $\alpha = \frac{1}{C} \downarrow 0$. An advantage of our parameterization is that we can generate different portfolios, with the same generating function $\varphi$, by varying $\alpha$ and $C$. Moreover, the $L^{(\alpha)}$-divergence is the usual $L$-divergence when $\alpha = 1$, and tends to the Bregman divergence when $\alpha \downarrow 0$. In our context, it is the canonical interpolation between the $L$-divergence and the Bregman divergence. By varying the parameters, we may interpolate between the two known portfolio constructions. We hope this work clarifies further the role of optimal transport and information geometry in functional portfolio generation. %

\subsection{Outline of the paper}
The remainder of the paper is organized as follows. In Section \ref{sec:market} we present the discrete time market model and introduce various ways of representing a trading strategy and the associated value process. In Section \ref{sec:fgp} we recall the definitions of multiplicatively and additively generated portfolios, and compare the corresponding decomposition formulas for the portfolio value process. These results motivate our general framework of functional portfolio generation. Section \ref{sec:geometry} explains the connections with optimal transport and information geometry. Our main results about functional portfolio generation and pathwise decomposition are proved in Section \ref{sec:main}. Finally, we illustrate the new portfolio construction with an empirical example. %

\section{The market model} \label{sec:market}
We work in a discrete time, model-independent framework that is used in our previous papers \cite{PW13, PW15}. The reader may refer to these papers for further details. Let $n \geq 2$, the number of stocks in the market, be fixed. The data of our model is a sequence $\{X(t) = (X_1(t), \ldots, X_n(t))\}_{t = 0}^{\infty}$ with values in $(0, \infty)^n$. We regard $X_i(t)$ as the market capitalization of stock $i$ at time $t$. The vector of market weights at time $t$ is defined by
\begin{equation} \label{eqn:market.weight}
\mu(t) = (\mu_1(t), \ldots, \mu_n(t)) = \left( \frac{X_1(t)}{X_1(t) + \cdots + X_n(t)}, \ldots, \frac{X_n(t)}{X_1(t) + \cdots + X_n(t)}\right).
\end{equation}
The vector $\mu(t)$ takes values in the open unit simplex
\begin{equation} \label{eqn:simplex}
\Delta_n = \{p = (p_1, \ldots, p_n) \in (0, 1)^n: p_1 + \cdots + p_n = 1\}.
\end{equation}
We denote by $\overline{\Delta}_n$ the closure of $\Delta_n$ in $\mathbb{R}^n$.

In this market we consider various self-financing trading strategies. Let us express a strategy in terms of the number of shares held at each point in time. Furthermore, we use the market portfolio as the numeraire. This means that the (relative) value of stock $i$ is simply the market weight $\mu_i(t)$. We assume that trading is frictionless (see \cite{PW13, PW15} for more details).

\begin{definition}[Trading strategy]
A self-financing trading strategy is a sequence $\eta = \{\eta(t)\}_{t = 0}^{\infty}$, with values in $\mathbb{R}^n$, such that the self-financing identity
\begin{equation} \label{eqn:self.financed}
\sum_{i = 1}^n \eta_i(t) \mu_i(t + 1) \equiv \sum_{i = 1}^n \eta_i(t + 1) \mu_i(t + 1)
\end{equation}
holds for all time $t$. The (relative) value process of $\eta$ is defined by
\begin{equation} \label{eqn:value}
V_{\eta}(t) = V_{\eta}(0) + \sum_{s = 0}^{t - 1} \left( \eta(t) \cdot (\mu(t + 1) - \mu(t)) \right),
\end{equation}
where $V_{\eta}(0) = \eta(0) \cdot \mu(0)$ and $a \cdot b$ is the Euclidean inner product.
\end{definition}

In the above definition, note that the portfolio's initial value is determined implicitly by $V_{\eta}(0) = \eta(0) \cdot \mu(0)$.

The portfolio $\eta(t)$ at time $t$ is chosen as a function of the previous prices $\{\mu(s)\}_{s = 0}^t$ up to time $t$ as well as other currently available information that can be modeled by a filtration $\{\mathcal{F}(t)\}_{t = 0}^{\infty}$. Nevertheless, we do not think of $\{\mu(t)\}_{t = 0}^{\infty}$ as a stochastic process. It is just some fixed path, chosen by nature, whose components are revealed one after another to the investor. In this paper we only study the value of a trading strategy relative to the market portfolio, so for simplicity we may omit the word `relative'. Note that because we allow both long and short positions in the portfolio, the value $V_{\eta}(t)$ may take negative values.

If the portfolio value $V_{\eta}(t)$ is strictly positive for all $t$, we may define the corresponding portfolio process by
\begin{equation} \label{eqn:portfolio.process}
\pi(t) = (\pi_1(t), \ldots, \pi_n(t)) = \left( \frac{\eta_1(t) \mu_1(t)}{V_{\eta}(t)}, \ldots, \frac{\eta_n(t) \mu_n(t)}{V_{\eta}(t)}\right).
\end{equation}
The components of $\pi(t)$ represent the percentages of current capital invested in each of the stocks; clearly $\sum_{i = 1}^n \pi_i(t) \equiv 1$. We call $\pi(t)$ the portfolio weight vector at time $t$. In this case, the value $V_{\eta}(t)$ can be expressed multiplicatively in the form
\begin{equation} \label{eqn:V.multiplicative}
V_{\eta}(t) = V_{\eta}(0) \prod_{s = 0}^{t - 1} \left(\pi(t) \cdot \frac{\mu(t + 1)}{\mu(t)}\right),
\end{equation}
where $\frac{\mu(t + 1)}{\mu(t)}$ is the vector of componentwise ratios. Compare this with the additive representation \eqref{eqn:value}. Intuitively, the definition of the value process given by \eqref{eqn:portfolio.process} and \eqref{eqn:V.multiplicative} limits the forms of functional portfolio generation that are compatible with this structure.

Before ending this section we recall an observation from \cite[Proposition 2.3]{KR17}. Given a sequence $\{\tilde{\eta}(t)\}_{t = 0}^{\infty}$ in $\mathbb{R}^n$ that may not be self-financing in the sense of \eqref{eqn:self.financed}, we can make a self-financing trading strategy out of it by setting
\begin{equation} \label{eqn:corrected.strategy}
\eta_i(t) = \tilde{\eta}_i(t) - Q^{\eta}(t) - C,
\end{equation}
where
\begin{equation} \label{eqn:defect}
Q^{\tilde{\eta}}(t) := \tilde{\eta}(t) \cdot \mu(t) - \tilde{\eta}(0) \cdot \mu(0) - \sum_{s = 0}^{t - 1} \tilde{\eta}(s) \cdot (\mu(s + 1) - \mu(s))
\end{equation}
is the `defect of self-financibility' and $C \in \mathbb{R}$ is a constant that controls the initial value $V_{\eta}(0) = \eta(0) \cdot \mu(0)$ of the portfolio. The `corrected' strategy $\eta$ satisfies
\begin{equation} \label{eqn:eta.identity}
\eta(t) \cdot (\mu(t + 1) - \mu(t)) = \tilde{\eta}(t) \cdot (\mu(t + 1) - \mu(t))
\end{equation}
for all $t$, so, remarkably, the returns are not affected. In \cite{KR17} the authors formulated these quantities in continuous time. The above definitions are straightforward adaptations to the present context.%

\section{Multiplicatively and additively generated portfolios} \label{sec:fgp}
In this section we review the definitions and main results of the two known forms of functional generation. For our purposes and for simplicity of exposition, we will assume that the generating functions are smooth and concave in a suitable sense.

\subsection{Multiplicatively generated portfolio}
We follow the treatment of \cite{PW15} which extended Fernholz's original construction.

\begin{definition}[Multiplicatively generated portfolio] \label{def:multiplicative.fgp}
Let $\varphi: \Delta \rightarrow \mathbb{R}$ be a smooth function such that $\Phi := e^{\varphi}$ is concave. (We call $\varphi$ an exponentially concave function.). Given the generating function $\varphi$, we define a mapping $\boldsymbol{\pi} : \Delta_n \rightarrow \overline{\Delta}_n$, called the portfolio map, by
\begin{equation} \label{eqn:multiplicative.map}
\boldsymbol{\pi}_i(p) = p_i \left(1 + D_{e_i - p} \varphi(p) \right), \quad i = 1, \ldots, n,
\end{equation}
where $(e_1, \ldots, e_n)$ is the standard Euclidean basis and $D_{e_i - p}$ is the directional derivative along the vector $e_i - p$. The portfolio process generated by $\varphi$ is given by
\begin{equation} \label{eqn:multiplicative.portfolio}
\pi(t) = \left(\frac{\eta_1(t) \mu_1(t)}{V_{\eta}(t)}, \ldots, \frac{\eta_n(t) \mu_n(t)}{V_{\eta}(t)}\right) = \boldsymbol{\pi}(\mu(t)).
\end{equation}
\end{definition}

Here is a geometric way (which is new but follows immediately from known results) to interpret the formula \eqref{eqn:multiplicative.map} which probably looks a bit strange at first sight. Consider the graph of the concave function $\Phi = e^{\varphi}$. Given $p \in \Delta_n$, let the tangent hyperplane to $\Phi$ at $p$ be given by $q \mapsto \sum_{i = 1}^n c_i q_i$ (see Figure \ref{fig:fgp1}). Then, it can be verified that the portfolio vector $\boldsymbol{\pi}(p)$ is given by
\begin{equation} \label{eqn:multiplicative.portfolio.2}
\boldsymbol{\pi}_i(p) = \frac{c_i p_i}{c_1p_1 + \cdots + c_np_n}, \quad i = 1, \ldots, n.
\end{equation}
In particular, the weight ratio $\frac{\boldsymbol{\pi}_i(p)}{p_i}$ is proportional to $c_i$. This construction is said to be multiplicative because we are specifying the weight ratios in terms of the derivatives of the generating function $\varphi$. The weight ratio can be regarded as an unnormalized Radon-Nikodym derivative of the portfolio weight with respect to the market weight.

\usetikzlibrary{positioning}
\tikzset{state/.style={circle,draw=black, very thick,inner sep=3pt,fill=white}}

\begin{figure}[t!]
\begin{tikzpicture}[scale = 0.55]
\draw (0, 0) to (5, 1);
\draw (5, 1) to (10, -1);
\draw (10, -1) to (0, 0);

\draw[dashed] (0, 0) to (0, 7);
\draw[dashed] (5, 1) to (5, 8);
\draw[dashed] (10, -1) to (10, 6);

\draw [draw=black, fill=gray, opacity=0.25] (0,5) -- (5,5) -- (10,2) -- cycle;

\path[fill=blue!50,opacity=.5] (1, 3) to [bend left=30] (4.2, 4.5) to [bend left=20] (8.2, 2) to [bend left=20] (1, 3);
\draw[blue] (1, 3) to [bend left=30] (4.2, 4.5);
\draw[blue] (4.2, 4.5) to [bend left=20] (8.2, 2);
\draw[blue] (8.2, 2) to [bend left = 20] (1, 3);
\draw[blue, dashed] (1, 3) to [bend left=10] (8.2, 2);

\node[circle, draw = black, fill = black,  inner sep=0pt, minimum size=3pt] at (0, 5) {};
\node[circle, draw = black, fill = black,  inner sep=0pt, minimum size=3pt] at (10, 2) {};
\node[circle, draw = black, fill = black,  inner sep=0pt, minimum size=3pt] at (5, 5) {};
\node [black, left] at (0, 5) {{\footnotesize $c_1$}};
\node [black, right] at (10, 2) {{\footnotesize $c_2$}};
\node [black, right] at (5, 5.1) {{\footnotesize $c_n$}};
\node [black, right] at (4.7, 0.06) {{\footnotesize $\Delta_n$}};

\node[circle, draw=blue, fill = black, inner sep=0pt, minimum size=3pt, label = below: {\footnotesize $\Phi(p)$}] at (4.2, 4.5) {};
\end{tikzpicture}
\caption{Geometric interpretation of multiplicatively generated portfolio.} \label{fig:fgp1}
\end{figure}
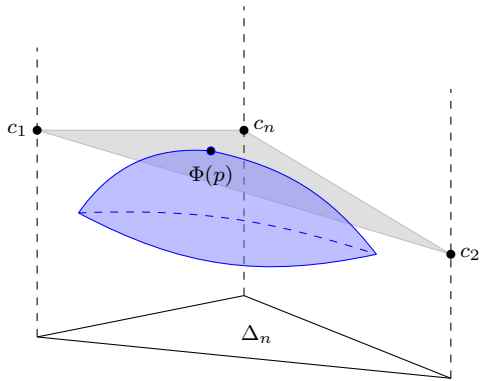

The reason why we need concave functions is that these are the only portfolio maps that capture volatility in a model-independent way. For precise statements of this assertion see \cite[Theorem 1]{PW15} and Section \ref{sec:geometry} below. To state the decomposition formula we also need the concept of $L$-divergence introduced in \cite{PW15}.

\begin{definition}[$L$-divergence] \label{def:L.divergence}
Let $\varphi$ be as in Definition \ref{def:multiplicative.fgp}. The $L$-divergence of $\varphi$ is the functional $D_{L, \varphi} \left[ \cdot \mid \cdot \right] : \Delta_n \times \Delta_n \rightarrow [0, \infty)$ defined by
\begin{equation} \label{eqn:L.divergence}
D_{L, \varphi}\left[q \mid p\right] = \log \left(1 + \nabla \varphi(p) \cdot (q - p) \right) - \left( \varphi(q) - \varphi(p) \right), \quad p, q \in \Delta_n,
\end{equation}
where $\nabla \varphi$ is the usual Euclidean gradient.
\end{definition}

Here $L$ stands for `logarithmic'. In \cite{PW15} we proved that the $L$-divergence is non-negative, and is strictly positive when $e^{\varphi}$ is strictly concave. Note that because the function $\varphi$ is the logarithm of a concave function, in \eqref{eqn:L.divergence} we can include a logarithmic correction which improves the standard linear approximation. This contrasts with the classical Bregman divergence reviewed in Definition \ref{def:Bregman} below.

Each and every multiplicatively generated portfolio admits a pathwise decomposition given as follows.

\begin{theorem} [Multiplicative decomposition] \cite{F99, PW15} \label{thm:fernholz}
Let $\eta$ be the trading strategy generated multiplicatively by the exponentially concave function $\varphi$ as in Definition \ref{def:multiplicative.fgp}. Then the value process of $\eta$ is given by
\begin{equation} \label{eqn:multiplicative.decomp}
\log V_{\eta}(t) - \log V_{\eta}(0) = \varphi(\mu(t)) - \varphi(\mu(0)) + \sum_{s = 0}^{t - 1} D_{L, \varphi} \left[ \mu(t + 1) \mid \mu(t)\right].
\end{equation}
\end{theorem}

The financial meaning of \eqref{eqn:multiplicative.decomp} will be discussed, in Section \ref{sec:financial.meaning}, after we introduce additively generated portfolios. Here let us give an example (for more examples see \cite[Chapter 3]{F02}). For a fixed element $\pi \in \overline{\Delta}_n$, let
\[
\varphi(p) = \sum_{i = 1}^n \pi_i \log p_i
\]
be the cross entropy in information theory (see \cite{CT06}), or equivalently the logarithm of the geometric mean $\Phi(p) = p_1^{\pi_1} \cdots p_n^{\pi_n}$ with weights $\pi$. Then the portfolio map $\boldsymbol{\pi}(\cdot)$ generated by $\varphi$ is simply the constant map $\boldsymbol{\pi}(p) \equiv \pi$, $p \in \Delta_n$. We call this the constant-weighted portfolio. The corresponding $L$-divergence is given by
\begin{equation} \label{eqn:free.energy}
D_{L, \varphi}\left[q \mid p\right] = \log \left( \pi \cdot \frac{q}{p} \right) - \pi \cdot \log \frac{q}{p}, \quad p, q \in \Delta_n.
\end{equation}
Here $\log \frac{q}{p}$ has components $\log \frac{q_i}{p_i}$. In finance this quantity is known as the excess growth rate \cite{F02, PW13} and the diversification return \cite{BF92, W11}.

 General comparison properties of $L$-divergences are studied in \cite{W15}. For example, it is proved there that if $\varphi(p) = \sum_{i = 1}^n \frac{1}{n} \log p_i$ which corresponds to the equal-weighted portfolio and $D_{\tilde{\varphi}, L}\left[ \cdot \mid \cdot \right] \geq D_{\varphi, L}\left[ \cdot \mid \cdot \right]$ for some $\tilde{\varphi}$, then $\varphi \equiv \tilde{\varphi}$ up to an additive constant. In other words, the $L$-divergence of the equal-weighted portfolio  is maximal and cannot be improved globally.

\subsection{Additively generated portfolio}
Now we adapt the work of \cite{KR17} and \cite{V16} to our discrete time setting, and derive a decomposition formula analogous to \eqref{eqn:multiplicative.decomp}.

The main idea of additively generated portfolio is the following. Take a (concave) generating function $\varphi: \Delta_n \rightarrow \mathbb{R}$ and consider its gradient
\begin{equation} \label{eqn:eta.gradient}
\tilde{\eta}(t) = \nabla \varphi(\mu(t)).
\end{equation}
The trading strategy generated additively by $\varphi$ is the self-financing strategy obtained from \eqref{eqn:eta.gradient} in the manner of \eqref{eqn:corrected.strategy} which corrects for the `defect of self-financibility'. The formula \eqref{eqn:additive.portfolio} we give below is taken from \cite[Section 3.3]{V16}.

\begin{definition} [Additively generated portfolio] \label{def:additive.portfolio}
Let $\varphi$ be a smooth concave function on $\Delta_n$. The trading strategy generated additively by $\varphi$ is defined by
\begin{equation} \label{eqn:additive.portfolio}
\eta_i (t) = D_{e_i - \mu(t)} \varphi(\mu(t)) + V_{\eta}(t), \quad i = 1, \ldots, n.
\end{equation}
This is well-defined for all $t$ once the initial value $V_{\eta}(0)$ is fixed.
\end{definition}

Note that the additively generated trading strategy depends on both the current market weight $\mu(t)$ and the current portfolio value $V_{\eta}(t)$, while the portfolio weights of the multiplicatively generated portfolio are deterministic functions of $\mu(t)$.

For completeness and to illustrate the notations, let us check the following

\begin{lemma} \label{lem:self.finance}
The trading strategy $\eta$ defined by \eqref{eqn:additive.portfolio} is self-financed, i.e., the identity \eqref{eqn:self.financed} holds.
\end{lemma}
\begin{proof}
Suppose the strategy is self-financing up to time $t$. Then, the portfolio value at time $t + 1$ is well defined and is equal to
\[
V_{\eta}(t + 1) = \sum_{i = 1}^n \eta_i(t) \mu_i(t + 1).
\]
On the other hand, we have
\begin{equation*}
\begin{split}
\sum_{i = 1}^n \eta_i(t + 1) \mu_i(t + 1) &= \sum_{i = 1}^n (D_{e_i - \mu(t + 1)} \varphi(\mu(t + 1)) + V_{\eta}(t + 1) ) \mu_i(t + 1) \\
&= \sum_{i = 1}^n \left( \mu_i(t + 1) D_{e_i - \mu(t + 1)} \varphi(\mu(t + 1))\right) + V_{\eta}(t + 1) \\
&= 0 + V_{\eta}(t + 1).
\end{split}
\end{equation*}
The last equality follows from the linearity of the directional derivative and the fact that
\[
\sum_{i = 1}^n \mu_i(t + 1) (e_i - \mu(t + 1)) = \mu(t + 1) - \mu(t + 1) = 0.
\]
By induction, the strategy $\eta$ is self-financing for all $t$.
\end{proof}

Now we develop a decomposition formula for the value process which was first established, in continuous time, in \cite{KR17}. In our discrete time setup, the finite variation term in \cite[Proposition 4.3]{KR17} becomes the accumulated sum of the Bregman divergence. Bregman divergence was introduced in \cite{B67} and is a fundamental concept in information geometry. The geometric consequences will be explained in Section \ref{sec:geometry}.

\begin{definition}[Bregman divergence] \label{def:Bregman}
The Bregman divergence of the concave function $\varphi$ in Definition \ref{def:additive.portfolio} is the non-negative functional $D_{B, \varphi}\left[ \cdot \mid \cdot \right] : \Delta_n \times \Delta_n \rightarrow [0, \infty)$ defined by
\begin{equation} \label{eqn:Bregman}
D_{B, \varphi}\left[q \mid p\right] = \nabla \varphi(p) \cdot (q - p) - \left( \varphi(q) - \varphi(p) \right).
\end{equation}
\end{definition}

\begin{theorem} [Additive decomposition] \label{thm:additive.decomp}
Consider the trading strategy $\eta$ generated additively by the concave function $\varphi$. Then the relative value satisfies the decomposition
\begin{equation} \label{eqn:additive.decomp}
V_{\eta}(t) - V_{\eta}(0) = \varphi(\mu(t)) - \varphi(\mu(0)) + \sum_{s = 0}^{t - 1} D_{B, \varphi}  \left[\mu(t + 1) \mid \mu(t) \right].
\end{equation}
\end{theorem}
\begin{proof}
For any $t$, we have
\begin{equation*}
\begin{split}
V_{\eta}(t + 1) - V_{\eta}(t) &= \sum_{i = 1}^n \eta_i (t) (\mu_i(t + 1) - \mu_i(t))\\
  &= \left( \sum_{i = 1}^n D_{e_i - \mu(t)} \varphi(\mu(t)) + V_{\eta}(t) \right) (\mu_i(t + 1) - \mu_i(t)) \\
  &= \sum_{i = 1}^n (\mu_i(t + 1) - \mu_i(t)) D_{e_i - \mu(t)} \varphi(\mu(t)) \\
  &= \nabla \varphi(\mu(t)) \cdot (\mu(t + 1) - \mu(t)) \\
  &= \varphi(\mu(t + 1)) - \varphi(\mu(t)) + D_{B, \varphi}\left[\mu(t + 1) \mid \mu(t)\right].
\end{split}
\end{equation*}
The formula \eqref{eqn:additive.decomp} is obtained by summing over time.
\end{proof}

We also give an example for the additive case. Let
\[
\varphi(p) = \frac{-1}{2} |p|^2 = \frac{-1}{2} \left(p_1^2 + \cdots + p_n^2\right)
\]
be half of the negative of the squared Euclidean norm. It generates the trading strategy $\eta(t)$ given by $
\eta_i(t) = |p|^2 - p_i + V_{\eta}(t)$. It is interesting to note that the Bregman divergence of $\varphi$ is half of the squared Euclidean distance:
\[
D_{B, \varphi}\left[q \mid p\right] = \frac{1}{2} \|p - q\|^2.
\]
For other examples of Bregman divergences see \cite[Chapter 1]{A16}.

\subsection{A framework of functional portfolio generation} \label{sec:financial.meaning}
Observe that both decompositions \eqref{eqn:multiplicative.decomp} and \eqref{eqn:additive.decomp} can be written in the form
\begin{equation} \label{eqn:general.decomp}
g(V_{\eta}(t)) - g(V_{\eta}(0)) = \varphi(\mu(t)) - \varphi(\mu(0)) + D\left[\mu(t + 1) \mid \mu(t) \right],
\end{equation}
where $g$, $\varphi$ and $D\left[ \cdot \mid \cdot\right]$ are suitable functions:
\begin{itemize}
\item (Multiplicative generation) $g(x) = \log x$ and $D\left[ \cdot \mid \cdot\right]$ is the $L$-divergence of the exponentially concave function $\varphi$.
\item (Additive generation) $g(x) = x$ and $D\left[ \cdot \mid \cdot\right]$ is the Bregman divergence of the concave function $\varphi$.
\end{itemize}

It is natural to ask if there exists other portfolio constructions that admit pathwise decompositions of the form \eqref{eqn:general.decomp}. To formulate this question we introduce the general concept of divergence. In information geometry, a divergence is a non-negative, distance-like quantity on a differentiable manifold; see \cite[Chapter 1]{A16}. Here we only consider divergences on the unit simplex.

\begin{definition}[Divergence on $\Delta_n$] \label{def:divergence}
A divergence on $\Delta_n$ is a non-negative functional $D[ \cdot : \cdot ] : \Delta_n \times \Delta_n \rightarrow [0, \infty)$ satisfying the following conditions:
\begin{itemize}
\item[(i)] $D[ q : p] = 0$ if and only if $p = q$.
\item[(ii)] It admits a quadratic approximation of the form
\begin{equation} \label{eqn:Riemannian}
D[p + \Delta p : p] = \frac{1}{2} \sum_{i, j = 1}^n g_{ij}(p) \Delta p_i \Delta p_j + O(| \Delta p |^3),
\end{equation}
where $|\Delta p| \rightarrow 0$, and the matrix $G(p) = \left( g_{ij}(p) \right)$ varies smoothly in $p$ and is strictly positive definite in the sense that
\begin{equation} \label{eqn:positive.definite}
\sum_{i, j = 1}^n g_{ij}(p) v_i v_j > 0
\end{equation}
for all vectors $v \in \mathbb{R}^n$ that are tangent to $\Delta_n$, i.e., $v_1 + \cdots + v_n = 0$.
\end{itemize}
If condition (i) is dropped and in \eqref{eqn:positive.definite} we do not strict inequality, we call $D\left[ \cdot \mid \cdot \right]$ a pseudo-divergence.
\end{definition}

Note that a divergence is not necessarily symmetric in the variables $p$ and $q$. In our context, this is intuitively clear because reversing the direction of time almost always leads to different financial outcomes. The $L$-divergence and the Bregman divergences are true divergences when the Euclidean Hessians of $\Phi = e^{\varphi}$ and $\varphi$ respectively are strictly positive definite (see \eqref{eqn:quadratic.metric} below for the matrix $g$).

\begin{definition} [General functional portfolio construction] \label{def:generation}
Let $\eta = \{\eta(t)\}_{t = 0}^{\infty}$ be a self-financing trading strategy whose value process is $V_{\eta}(t)$, and let $\varphi, g: \Delta_n \rightarrow \mathbb{R}$ be functions on $\Delta_n$ where $g$ is strictly increasing. We say that  $\eta$ is generated by $\varphi$ with scale function $g$ if there exists a pseudo-divergence $D[\cdot  : \cdot ]$ on $\Delta_n$ such that \eqref{eqn:general.decomp} holds for all market sequences $\{\mu(t)\}_{t = 0}^{\infty}$.
\end{definition}

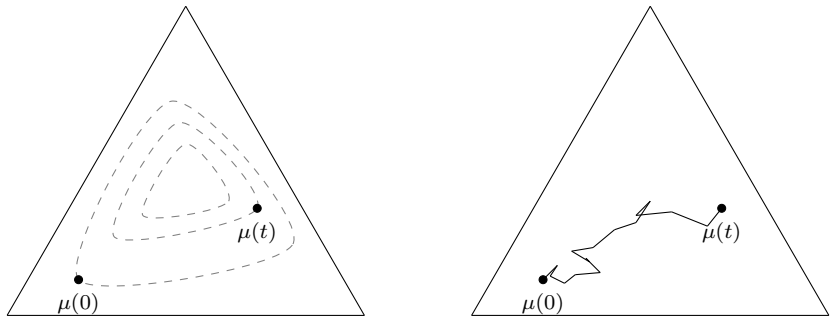
\begin{figure}[t!]
\begin{tikzpicture}[scale = 0.95]
\draw (0, 0) to (5, 0);
\draw (5, 0) to (2.5, 4.33);
\draw (2.5, 4.33) to (0, 0);

\draw [gray, dashed] plot [smooth cycle] coordinates {(1, 0.5) (4,1) (2.3, 3)};
\draw [gray, dashed] plot [smooth cycle] coordinates {(1.5, 1.1) (3.5, 1.5) (2.4, 2.7)};
\draw [gray, dashed] plot [smooth cycle] coordinates {(1.9, 1.4) (3.1, 1.6) (2.5, 2.4)};

\node[circle, draw=black, fill = black, inner sep=0pt, minimum size=3pt, label = below: {\footnotesize $\mu(0)$}] at (1, 0.5)  {};
\node[circle, draw=black, fill = black, inner sep=0pt, minimum size=3pt, label = below: {\footnotesize $\mu(t)$}] at (3.5, 1.5) {};

\draw (6.5, 0) to (11.5, 0);
\draw (11.5, 0) to (9, 4.33);
\draw (9, 4.33) to (6.5, 0);
\node[circle, draw=black, fill = black, inner sep=0pt, minimum size=3pt, label = below: {\footnotesize $\mu(0)$}] at (7.5, 0.5)  {};
\node[circle, draw=black, fill = black, inner sep=0pt, minimum size=3pt, label = below: {\footnotesize $\mu(t)$}] at (10, 1.5) {};

\draw (7.5, 0.5) to (7.7, 0.7);
\draw (7.7, 0.7) to (7.6, 0.55);
\draw (7.6, 0.55) to (7.8, 0.45);
\draw (7.8, 0.45) to (7.95, 0.57);
\draw (7.95, 0.57) to (8.3, 0.6);
\draw (8.3, 0.6) to (8.1, 0.8);
\draw (8.1, 0.8) to (8.15, 0.75);
\draw (8.15, 0.75) to (7.9, 0.9);
\draw (7.9, 0.9) to (8.2, 0.95);
\draw (8.2, 0.95) to (8.5, 1.2);
\draw (8.5, 1.2) to (8.8, 1.3);
\draw (8.8, 1.3) to (9, 1.6);
\draw (9, 1.6) to (8.8, 1.4);
\draw (8.8, 1.4) to (9.3, 1.45);
\draw (9.3, 1.45) to (9.8, 1.25);
\draw (9.8, 1.25) to (10, 1.5);


\end{tikzpicture}
\caption{Financial intuition of a functionally generated trading strategy. Left: Displacement of the market where the dotted curves are level sets of the generating function $\varphi$. Right: Volatility of the path $\{\mu(s)\}_{s = 0}^t$ as an accumulated divergence.}
\label{fig:decomp}
\end{figure}

We will characterize functional generation according to Definition \ref{def:generation} in Section \ref{sec:main}. Before closing this section we discuss briefly the financial intuition behind the decomposition \eqref{eqn:general.decomp}. The decomposition implies that the performance of the portfolio relative to the market can be attributed to two quantities. The first is the displacement of the market measured by the change of the function $\varphi(\mu(t))$. It depends only on the beginning location $\mu(0)$ and the current location $\mu(t)$ of the market. Note that change in $\varphi(\mu(t))$ is only caused by the component of market movement along the direction of $\nabla \varphi(\mu(t))$ which is perpendicular to the level set. In particular, a displacement along the same level set of $\varphi$ is not visible in this first term. The second term in \eqref{eqn:general.decomp} measures the volatility of the market, as it travels from $\mu(0)$ to $\mu(t)$, measured by the sum of $D\left[ \mu(s + 1) \mid \mu(s) \right]$ over time. Intuitively, the functionally generated trading strategy $\eta$ outperforms the market if and only if the volatility is greater than the displacement of the market. In this spirit, relative arbitrages can be constructed by imposing conditions on $\varphi(\mu(t))$ and the growth of the divergence term.

In the continuous time limit, the divergence term $\sum_{s = 0}^{t - 1} D\left[ \mu(s + 1) \mid \mu(s) \right]$ becomes a finite variation process given in terms of the quadratic variation process of the market weights. This follows from the quadratic approximation \eqref{eqn:Riemannian}. In particular, we have, in matrix notations,
\begin{equation} \label{eqn:quadratic.metric}
\begin{split}
D_{B, \varphi}\left[p + \Delta p \mid p\right] &= \frac{-1}{2} (\Delta p)^{\top} \Hess \varphi(p) (\Delta p) + O(|\Delta p|^3), \\
D_{L, \varphi}\left[p + \Delta p \mid p \right] &= \frac{-1}{2} (\Delta p)^{\top} \left( \Hess \varphi(p) + (\nabla \varphi(p))(\nabla \varphi(p) )^{\top}\right) (\Delta p) + O(|\Delta p|^3),
\end{split}
\end{equation}
where $\Hess \varphi$ is the Euclidean Hessian matrix. This gives \cite[Proposition 4.3]{KR17} and \cite[Theorem 3.1.5]{F02} in the continuous time setting.

\section{Connections with optimal transport and information geometry} \label{sec:geometry}
\subsection{Optimal transport}
Optimal transport is an extremely vast and deep field. For simplicity, we explain the main ideas using transport of discrete masses. (As shown in \cite{PW15} and \cite{V16} the results hold for the Monge-Kantorovich problem with general probability measures satisfying mild conditions.) For systematic expositions of optimal transport we refer the reader to the excellent books \cite{V03, V08} by Villani.

Let ${\mathcal{X}}$, ${\mathcal{Y}}$ be sets, and let $c: {\mathcal{X}} \times {\mathcal{Y}} \rightarrow \mathbb{R}$ be a cost function. Let $x_1, \ldots, x_N \in \mathcal{X}$ and $y_1, \ldots, y_N \in \mathcal{Y}$. In this simplified context, the (Monge) optimal transport problem attempts to  assign each $x_i$ to some $y_j$ so as to minimize the total transport cost:
\[
\min \sum_{i = 1}^N c(x_i, y_{\sigma(i)}).
\]
Here the minimization is over all bijections (i.e., permutations) $\sigma: \{1, \ldots, N\} \rightarrow \{1, \ldots, N\}$. Using the fact that a permutation is a product of disjoint cycles, it is not difficult to show (see \cite[Exercise 2.21]{V03}) that the assignment $x_i \rightarrow y_i$, $i = 1, \ldots, N$, is optimal if and only if it is $c$-cyclical monotone, i.e., for any $i_1, i_2, \ldots, i_m$ in $\{1, \ldots, N\}$ we have
\begin{equation} \label{eqn:ccm}
\sum_{k = 1}^m c(x_{i_k}, y_{i_k}) \leq \sum_{k = 1}^m c(x_{i_k}, y_{i_{k + 1}}),
\end{equation}
where by convention $i_{m + 1} := i_1$. Similarly, we can define the $c$-cyclical monotonicity of a subset $A \subset {\mathcal{X}} \times {\mathcal{Y}}$. Under very mild conditions (see \cite{V03, V08}), we can show that $c$-cyclical monotonicity of the transport plan is a sufficient and necessary condition for optimality in the Monge-Kantorovich problem.

The connection between functional portfolio generation and optimal transport is the following. Consider the decompositions \eqref{thm:fernholz} and \eqref{thm:additive.decomp}. If the market path satisfies $\mu(t) = \mu(0)$, so that $\{\mu(s)\}_{s = 0}^t$ is a cycle in $\Delta_n$, then we always have $V_{\eta}(t) \geq V_{\eta}(0)$. It turns out that this inequality is equivalent to $c$-cyclical monotonicity for appropriate choices of the cost function. This implies that functionally generated trading strategies can be obtained as solutions to certain optimal transport problems. The following theorem is taken from \cite{PW15} and \cite{V16}.

\begin{theorem}[Functional generation and optimal transport] \label{thm:transport} {\ }
\begin{enumerate}
\item[(i)] (Multiplicative generation \cite{PW15}) Let ${\mathcal{X}} = \Delta_n$ and ${\mathcal{Y}} = \overline{\Delta}_n$, and consider the cost function
\begin{equation} \label{eqn:multiplicative.cost}
c(p, q) = \log \left(p \cdot q \right), \quad (p, q) \in {\mathcal{X}} \times {\mathcal{Y}}.
\end{equation}
Suppose $\boldsymbol{\pi}$ is the portfolio map generated multiplicatively by the exponentially concave function $\varphi$ via \eqref{eqn:multiplicative.map}. Given $\boldsymbol{\pi}$, define the map $T: \Delta_n \rightarrow \overline{\Delta}_n$ by
\begin{equation} \label{eqn:multiplicative.transport}
T(p) = \left( \frac{\boldsymbol{\pi}_1(p)/p_1}{\sum_{j = 1}^n \boldsymbol{\pi}_j(p) / p_j}, \ldots, \frac{\boldsymbol{\pi}_n(p)/p_n}{\sum_{j = 1}^n \boldsymbol{\pi}_j(p) / p_j}\right).
\end{equation}
Then the graph of $T$ is $c$-cyclical monotone.
\item[(ii)] (Additive generation  \cite{V16}) Consider instead ${\mathcal{X}} = \Delta_n$ and ${\mathcal{Y}} = \mathbb{R}^n$, and the cost function
\begin{equation} \label{eqn:additive.cost}
c(p, v) = p \cdot v, \quad (p, v) \in {\mathcal{X}} \times {\mathcal{Y}}.
\end{equation}
Consider the trading strategy $\eta$ generated additively by the concave function $\varphi$ via \eqref{eqn:additive.portfolio}. Define the transport map $T: \Delta_n \rightarrow \mathbb{R}^n$ by
\begin{equation}
T(p) = \nabla \varphi(p).
\end{equation}
Then the graph of $T$ is $c$-cyclical monotone.
\end{enumerate}
\end{theorem}

Let us note that in Theorem \ref{thm:transport}(ii), once the cost function is identified, the statement about $c$-cyclical monotonicity follows immediately from Rockafellar's theorem (see \cite[Section 24]{R70}) which characterizes the differentials of convex/concave functions.

\subsection{Geometry of additively generated portfolio}
That functionally generated portfolio is related to information geometry was first discovered in \cite{PW16} for multiplicatively generated portfolios. In that paper we studied the geometric structure (in the sense of \cite[Chapter 6]{A16}) induced by a given $L$-divergence of an exponentially concave function on the simplex $\Delta_n$ and proved a generalized Pythagorean theorem. Here we give the analogue for additive generation; this case is significantly easier as the geometry of Bregman divergence has been well studied.

Let $\eta$ be the trading strategy generated additively by a smooth concave function $\varphi$ on $\Delta_n$, so that $\Hess \varphi$ is strictly positive definite. In the market model discussed in Section \ref{sec:market} the portfolio is traded at every time point. More generally, we may trade according to an increasing sequence of stopping times. An important but challenging financial problem is to determine the optimal rebalancing frequency. For the case of three time points we can characterize the solution geometrically.

Let $t_0 < t_1 < t_2$ and consider two ways of implementing the strategy $\eta$:
\begin{itemize}
\item[(a)] Rebalance only at time $t_0$.
\item[(b)] Rebalance at times $t_1$ and $t_2$.
\end{itemize}
Then, at time $t_2$, the corresponding portfolio values are
\begin{equation*}
\begin{split}
V^{(a)}_{\eta}(t_2) - V^{(a)}_{\eta}(t_0) &= \varphi(\mu(t_2)) - \varphi(\mu(t_0)) + D\left[ \mu(t_2) \mid \mu(t_0)\right],\\
V^{(b)}_{\eta}(t_2) - V^{(b)}_{\eta}(t_0) &= \varphi(\mu(t_2)) - \varphi(\mu(t_1)) +  D\left[ \mu(t_2) \mid \mu(t_1)\right] +  D\left[ \mu(t_1) \mid \mu(t_0)\right],
\end{split}
\end{equation*}
where $D = D_{B, \varphi}$ is the Bregman divergence of $\varphi$. Assuming $V_{\eta}^{(a)}(t_0) = V_{\eta}^{(b)}(t_0)$, we see that method (b) earns more than method (a) over the period $[t_0, t_2]$ if and only if
\begin{equation} \label{eqn:pyth}
D\left[ \mu(t_2) \mid \mu(t_1)\right] +  D\left[ \mu(t_1) \mid \mu(t_0)\right] \geq D\left[\mu(t_2) \mid \mu(t_0)\right].
\end{equation}

We give a geometric condition on the triplet $(p, q, r) = (\mu(t_0), \mu(t_1), \mu(t_2)) \in \left( \Delta_n  \right)^3$ for \eqref{eqn:pyth} to hold. This will be stated in terms of the dually flat geometry of the Bregman divergence $\varphi$; for more details and motivations see \cite[Chapter 1]{A16}.

\begin{definition}[Geometry of Bregman divergence]
Let $M = \Delta_n$ considered as a smooth $(n - 1)$-dimensional smooth manifold without boundary.
\begin{enumerate}
\item[(i)] (Primal coordinate) The identity map $p = (p_1, \ldots, p_n) \in M \mapsto p = (p_1, \ldots, p_n)$ is called the primal coordinate system.
\item[(ii)] (Dual coordinate) The map $p \in M \mapsto p^* = \nabla \varphi(p)$ is called the dual coordinate system. By properties of the Legendre transform this is a diffeomorphism from $M$ to its range which is also convex. The inverse map is given by $p = \nabla \varphi^*(p^*)$ where $\varphi^*$ is the concave conjugate of $\varphi$.
\item[(iii)] (Primal geodesic) Given $p$ and $q$, the primal geodesic from $q$ to $p$ is the straight line
\[
\gamma(t) = (1 - t) q + t p
\]
when expressed in primal coordinates.
\item[(iv)] (Dual geodesic) Given $q$ and $r$, the dual geodesic from $q$ to $r$ is the straight line
\[
\gamma^*(t) = (1 - t) q^* + t r^*
\]
when expressed in dual coordinates.
\item[(v)] (Riemannian metric) Let $u$ and $v$ be tangent to $\Delta_n$ and expressed in primal coordinates, and let $p \in \Delta_n$. Their Riemannian inner product at $p$ is defined by
    \begin{equation} \label{eqn:inner.product}
    \langle u, v \rangle_p = u^{\top} \left(-\Hess \varphi(p)\right) v.
    \end{equation}
    We say that $u$ and $v$ are orthogonal at $p$ if $\langle u, v \rangle_p = 0$.
\end{enumerate}
\end{definition}

We say that this geometry is dually flat because the straight lines in the $p$ space and those in the $p^*$ space give two affine structures on $M$; the two coordinate systems are dual to each other through the Legendre transform. In \cite{PW16}, this dual structure is extended to the logarithmic cost function where the Legendre transform takes the form of \eqref{eqn:multiplicative.transport}.

\begin{figure}[t!]
\centering
\begin{tikzpicture}[scale = 0.5]
\draw[thick, blue] (0, 0) to[out=200, in=-30] (-10, 1);
\draw[thick, red] (-10, 1) to (-6, 6);
\node [right] at (0, 0) {$r$};
\node [left] at (-10, 1) {$q$};
\node [above] at (-6, 6) {$p$};
\node [blue, below] at (-5, -1) {dual geodesic};
\node [red, left] at (-7.7, 4) {primal geodesic};

\draw[->, thick] (-8, 2) to (-9.7, 1.05);
\node [right] at (-8, 2) {{\footnotesize $\perp$ under the Riemannian metric}};

\end{tikzpicture}
\caption{Generalized Pythagorean theorem. This is drawn in primal coordinates, so that the primal geodesic is a straight line.}
\label{fig:pyth}
\end{figure}
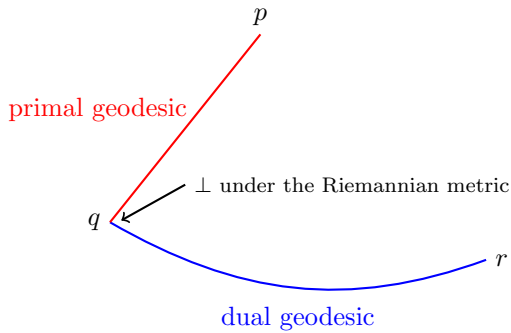

\begin{theorem}[Generalized Pythagorean theorem]
Let $p$, $q$ and $r$ be distinct points in $\Delta_n$. Consider the primal geodesic $\gamma(t)$ from $q$ to $p$ and the dual geodesic $\gamma^*(t)$ from $q$ to $r$. Then the inequality \eqref{eqn:pyth} holds if and only if the Riemannian angle between the velocity vectors $\dot{\gamma}(0)$ and $\dot{\gamma^*}(0)$ at $q$ is less than $\frac{\pi}{2}$. Moreover, equality holds if and only if the two geodesics meet orthogonally at $q$ (see Figure \ref{fig:pyth}).
\end{theorem}
\begin{proof}
See \cite[Theorem 1.2]{A16}.
\end{proof}

\section{Characterizing functional portfolio generation} \label{sec:main}
In this section we characterize functional portfolio generations according to Definition \ref{def:generation}. Throughout this section we let $\eta$ be a functionally generated trading strategy as in Definition \ref{def:generation}. We assume that the scale function $g$ is smooth and $g'(x) > 0$ for all $x$. We also require that the domain of $g$ contains the positive real line $(0, \infty)$. Furthermore, we assume that $\eta$ is non-trivial in the sense that for all paths of the market up to time $t$, the profit or loss
\[
V_{\eta}(t + 1) - V_{\eta}(t) = \eta(t) \cdot (\mu(t + 1) - \mu(t))
\]
is not identically zero as a function of $\mu(t + 1)$.

\subsection{The scale function}
We first characterize the possible shapes of the scale function.

\begin{proposition} \label{prop:scale.function}
The scale function $g$ admits a non-trivial functionally generated trading strategy only if it has one of the following forms. Either
\begin{equation} \label{eqn:shape1}
g(x) = c_1 x + c_2
\end{equation}
where $c_1 > 0$ and $c_2 \in \mathbb{R}$, or
\begin{equation} \label{eqn:shape2}
g(x) = c_2 \log(c_1 + x) + c_3
\end{equation}
where $c_1 \geq 0$, $c_2 > 0$ and $c_3 \in \mathbb{R}$.
\end{proposition}

We will prove Proposition \ref{prop:scale.function} with several lemmas. First we observe that the decomposition \eqref{eqn:general.decomp}  already implies a formula for the trading strategy at each point in time.

\begin{lemma}
For any tangent vector $v$ of $\Delta_n$ (i.e., $v_1 + \cdots + v_n = 0$) we have
\begin{equation} \label{eqn:eta.as.gradient}
\eta(t) \cdot v  =  \frac{1}{g'(V_{\eta}(t))} \nabla \varphi(\mu(t)) \cdot  v.
\end{equation}
In particular, for each $i = 1, \ldots, n$, we have
\begin{equation} \label{eqn:eta.formula}
\eta_i(t) = \frac{1}{g'(V_{\eta}(t))} D_{e_i - \mu(t)} \varphi(\mu(t)) + V_{\eta}(t).
\end{equation}
\end{lemma}
\begin{proof}
From \eqref{eqn:general.decomp} we have the identity
\begin{equation} \label{eqn:1step.decomp}
g(V_{\eta}(t + 1)) - g(V_{\eta}(t)) = \varphi(\mu(t + 1)) - \varphi(\mu(t)) + D\left[ \mu(t + 1) \mid \mu(t) \right]
\end{equation}
which holds for all values of $\mu(t + 1)$. Moreover, we have
\[
V_{\eta}(t + 1) = V_{\eta}(t) + \eta(t) \cdot (\mu(t + 1) - \mu(t)).
\]
Now let $\mu(t + 1) - \mu(t) = \delta v$, $\delta > 0$, and compute the first order approximation of both sides of \eqref{eqn:1step.decomp}. Since $D\left[ \cdot \mid \cdot \right]$ is a pseudo-divergence, by \eqref{eqn:Riemannian} its first order approximation vanishes. Evaluating the derivatives and dividing by $\delta > 0$, we obtain \eqref{eqn:eta.as.gradient}.

Letting $v = e_i - \mu(t)$ in \eqref{eqn:eta.formula} for $i = 1, \ldots, n$, we get the explicit formula \eqref{eqn:eta.formula}.
\end{proof}

Observe that \eqref{eqn:eta.as.gradient} reduces to \eqref{eqn:eta.gradient} when $g(x) = x$, and to \eqref{eqn:multiplicative.map} when $g(x) = \log x$. Also, we note that the trading strategy depends only on $\mu(t)$ and the current value $V_{\eta}(t)$. Putting $v = \mu(t + 1) - \mu(t)$ in \eqref{eqn:eta.as.gradient}, we have
\begin{equation} \label{eqn:for.use}
V_{\eta}(t + 1) - V_{\eta}(t) = \frac{1}{g'(V_{\eta}(t))}\nabla \varphi(\mu(t)) \cdot (\mu(t + 1) - \mu(t)).
\end{equation}

Consider the expression
\begin{equation} \label{eqn:gV.expression}
\begin{split}
& g(V_{\eta}(t + 1)) - g(V_{\eta}(t))\\
&= g(V_{\eta}(t) + (V_{\eta}(t + 1) - V_{\eta}(t))) - g(V_{\eta}(t)) \\
  &= g\left(V_{\eta}(t) + \frac{1}{g'(V_{\eta})}  \nabla \varphi(\mu(t)) \cdot ( \mu(t + 1) - \mu(t))\rangle)\right) - g(V_{\eta}(t)).
\end{split}
\end{equation}
By \eqref{eqn:general.decomp}, this equals
\[
\varphi(\mu(t + 1)) - \varphi(\mu(t)) + D[\mu(t + 1) : \mu(t)],
\]
which is a function of $\mu(t)$ and $\mu(t + 1)$ only. Thus, the expression in \eqref{eqn:gV.expression} does not depend on the current portfolio value $V_{\eta}(t)$. From this observation, we will derive a differential equation satisfied by $g$.

\begin{lemma}
The scale function $g$ satisfies the third order nonlinear ODE
\begin{equation} \label{eqn:ODE1}
g' g''' =2  (g'')^2.
\end{equation}
\end{lemma}
\begin{proof}
Let $x = V_{\eta}(t)$, and let $\delta =  \nabla \varphi(\mu(t)) \cdot ( \mu(t + 1) - \mu(t) )$. From \eqref{eqn:gV.expression}, for any $\delta$, the expression
\begin{equation} \label{eqn:invariance}
g(x + \frac{1}{g'(x)} \delta) - g(v)
\end{equation}
does not depend on $x$.

Differentiating \eqref{eqn:invariance} with respect to $x$, we have
\[
g'(x + \frac{1}{g'(x)}\delta) \left(1 - \delta \frac{g''(x)}{(g'(x))^2}\right) - g'(x) = 0.
\]
Next we differentiate with respect to $\delta$ (since $\eta$ is assumed to be non-trivial, this can be done by varying $\mu(t + 1)$):
\[
g''(x + \frac{1}{g'(x)} \delta) \frac{1}{g'(x)} \left(1 - \delta \frac{g''(x)}{(g'(x))^2}\right) + g'(x + \frac{1}{g'(x)} \delta) \frac{- g''(x)}{(g'(x))^2} = 0.
\]
Differentiating one more time with respect to $\delta$, we have
\begin{equation*}
\begin{split}
& g'''(x + \frac{1}{g'(x)}\delta) \frac{1}{(g'(x))^2} \left(1 - \delta \frac{g''(x)}{(g'(x))^2}\right)\\
&\quad + g''(x + \frac{1}{g'(x)} \delta) \frac{-g''(x)}{(g'(x))^3} - \frac{g''(x + \frac{1}{g'(x)} \delta) g''(x)}{(g'(x))^3} = 0
\end{split}
\end{equation*}
Setting $\delta = 0$, we get $\frac{g'''(x)}{(g'(x))^2} - 2 \frac{(g''(x))^2}{(g'(x))^3} = 0$. Rearranging, we obtain the ODE \eqref{eqn:ODE1}
\end{proof}

With the differential equation in hand, it is not difficult to check the following

\begin{lemma}
All solutions to the ODE \eqref{eqn:ODE1} can be written in the form
\begin{equation} \label{eqn:g.equation}
g(x) = c_0 + c_1 x  \quad \text{or} \quad  g(x) = c_2 \log (c_1 + x) + c_3,
\end{equation}
where the $c_i$'s are real constants.
\end{lemma}

\begin{proof}[Proof of Proposition \ref{prop:scale.function}]
The proposition is now a direct consequence of the previous lemmas. The conditions on the constants $c_i$ ensure that $g$ is strictly increasing and its domain contains $(0, \infty)$.
\end{proof}

\subsection{A new functional generation}
Note that scale functions of the form $g(x) = c_0 + c_1x$ correspond to additive generation, so they do not lead to new portfolios. Next, we will show that each scale function $g(x) = c_2 \log (c_1 + x) + c_3$ admits a functional portfolio generation. Note that the additive constant $c_3$ plays no role and may be omitted. Let us write $c_2 = \frac{1}{\alpha}$ and $c_1 = C$, where $\alpha > 0$ and $C \geq 0$ are the parameters. Motivated by \eqref{eqn:eta.formula} we give the following definition.

\begin{definition}[$(\alpha, C)$-generation] \label{def:genera.generation}
Let $\varphi: \Delta_n \rightarrow \infty$ be a smooth function such that $e^{\alpha \varphi}$ is  concave on $\Delta_n$. (We call $\varphi$ an $\alpha$-exponentially concave function.) Also let $V_{\eta}(t) > 0$ be fixed. The trading strategy $(\alpha, C)$-generated by $\varphi$ is the process $\eta(t)$ defined by
\begin{equation} \label{eqn:general.generation}
\eta_i(t) = \alpha (C + V_{\eta}(t)) D_{e_i - \mu(t)} \varphi(\mu(t)) + V_{\eta}(t), \quad i = 1, \ldots, n.
\end{equation}
Here the initial value $V_{\eta}(0)$ is fixed and arbitrary.
\end{definition}

It is easy to verify, in the manner of Lemma \ref{lem:self.finance}, that \eqref{eqn:general.generation} is a self-financed trading strategy. Comparing \eqref{eqn:general.generation} with \eqref{eqn:multiplicative.map} and \eqref{eqn:additive.portfolio}, we see that multiplicative generation corresponds to the case $C = 0$ and $\alpha = 1$, and additive generation corresponds to the limit when $\alpha = \frac{1}{C} \rightarrow 0$. Note that when $C \neq 0$, the portfolio weights depend on both $V_{\eta}(t)$ and $\mu(t)$. 

The trading strategy can be interpreted as follows. By increasing $C$, we may construct portfolios that are more aggressive than the multiplicatively generated portfolio. Note that we keep the parameter $\alpha$ so that we can generate different portfolios with the same generating function $\varphi$.

\begin{lemma} [Portfolio weight of $\eta$]
Let $\pi^{(\alpha)}$ be the portfolio process generated multiplicatively by the exponentially concave function $\alpha \varphi$. If $V_{\eta}(t) > 0$, the portfolio weight vector $\pi(t)$ of the $(\alpha, C)$-generated trading strategy $\eta$ is given by
\begin{equation} \label{eqn:eta.weight}
\pi(t) = \left( \frac{\eta_1(t) \mu_1(t)}{V_{\eta}(t)}, \ldots, \frac{\eta_n(t) \mu_n(t)}{V_{\eta}(t)}\right) = \frac{C + V_{\eta}(t)}{V_{\eta}(t)} \pi^{(\alpha)}(t) - \frac{C}{V_{\eta}(t)} \mu(t).
\end{equation}
In particular, at each point in time, $\eta(t)$ longs the multiplicatively generated portfolio $\pi^{(\alpha)}$ and shorts the market portfolio with weights depending on $V_{\eta}(t)$ and $C$.
\end{lemma}
\begin{proof}
Direct computation using \eqref{eqn:general.generation}.
\end{proof}

To formulate a pathwise decomposition, we need a divergence on $\Delta_n$ that corresponds to our $(\alpha, C)$-generation. This divergence is a generalization of both the Bregman and $L$-divergences.

\begin{definition}[$L^{(\alpha)}$-divergence]
Let $\varphi$ be differentiable and $\alpha$-exponentially concave on $\Delta_n$. The $L^{(\alpha)}$-divergence of $\varphi$ is the functional $D_{L^{(\alpha)}, \varphi}\left[ \cdot \mid \cdot\right]: \Delta_n \times \Delta_n \rightarrow [0, \infty)$ defined by
\begin{equation} \label{eqn:L.alpha.divergence}
D_{L^{(\alpha)}, \varphi}\left[q \mid p\right] = \frac{1}{\alpha} \log \left(1 + \alpha \nabla \varphi(p) \cdot (q - p) \right) + \left( \varphi(q) - \varphi(p) \right).
\end{equation}
In particular, it equals $\frac{1}{\alpha}$ times the usual $L$-divergence of $\alpha \varphi$, i.e,
\[
D_{L^{\alpha}, \varphi} \left[ \cdot \mid \cdot\right] \equiv \frac{1}{\alpha} D_{L^{(1)}, \alpha \varphi} \left[ \cdot \mid \cdot\right].
\]
\end{definition}

\begin{lemma}
Suppose $\varphi$ is $\alpha$-exponentially concave for all $\alpha > 0$ sufficiently small. Then, as $\alpha \rightarrow 0$, the $L^{(\alpha)}$-divergence converges pointwise to the Bregman divergence:
\begin{equation*}
\lim_{\alpha \downarrow 0} D_{L^{(\alpha)}, \varphi}\left[q \mid p\right] = D_{B, \varphi}\left[q \mid p\right], \quad p, q \in \Delta_n.
\end{equation*}
\end{lemma}
\begin{proof}
This follows from the limit $\frac{1}{\alpha} \log (1 + \alpha x) \rightarrow x$ when $\alpha \downarrow 0$.
\end{proof}

Now we derive the promised pathwise decomposition formula for an $(\alpha, C)$-generated trading strategy. Since the portfolio $\eta$ may have short positions, unfortunately we cannot guarantee that $g(V_{\eta}(t)) = \frac{1}{\alpha}\log (C + V_{\eta}(t))$ is well-defined for all $t$. Nevertheless, the decomposition holds whenever the value is bounded below by $-C$.

\begin{theorem}[Pathwise decomposition] \label{thm:general.decomp}
Consider an $(\alpha, C$)-generated trading strategy $\eta$ as in Definition \ref{def:genera.generation}. If $V_{\eta}(\cdot) > -C$, the value process satisfies the pathwise decomposition
\begin{equation} \label{eqn:new.decomp}
\frac{1}{\alpha} \log \frac{C + V_{\eta}(t)}{C + V_{\eta}(0)} = \varphi(\mu(t)) - \varphi(\mu(0)) + \sum_{s = 0}^{t - 1} D_{L^{(\alpha)}, \varphi}\left[ \mu(s + 1) \mid \mu(s) \right].
\end{equation}
\end{theorem}
\begin{proof}
By \eqref{eqn:general.generation}, at each time $t$ we have
\begin{equation*}
\begin{split}
& \frac{1}{\alpha} \log (C + V_{\eta}(t + 1)) - \frac{1}{\alpha} \log (C + V_{\eta}(t))\\
 &= \frac{1}{\alpha} \log \frac{C + V_{\eta}(t) + \alpha (C + V_{\eta}(t)) \nabla \varphi(\mu(t)) \cdot (\mu(t + 1) - \mu(t))}{C + V_{\eta}(t)} \\
 &= \frac{1}{\alpha} \log \left(1 + \alpha \nabla \varphi(\mu(t)) \cdot (\mu(t + 1) - \mu(t))\right) \\
 &= \varphi(\mu(t + 1)) - \varphi(\mu(t)) + D_{L^{(\alpha)}, \varphi}\left[\mu(t + 1) \mid \mu(t)\right].
\end{split}
\end{equation*}
This yields the desired decomposition.
\end{proof}

Note that in the above proof the term $V_{\eta}(t)$ nicely cancels out when we take the difference $g(V_{\eta}(t + 1)) - g(V_{\eta}(t))$. This is because the scale function satisfies the differential equation \eqref{eqn:ODE1}.

Since the $L^{(\alpha)}$-divergence is nothing but a normalized $L$-divergence, we may apply directly the information geometry developed in \cite{PW16} to state a Pythagorean theorem which characterizes optimal rebalancing for three time points. More interestingly, for a given exponentially concave function the $L^{(\alpha)}$-divergence introduces a natural interpolation between the $L$-divergence (when $\alpha = 1$) and the Bregman divergence (when $\alpha = 0$). In terms of optimal transport, we may interpolate between the logarithmic cost function and the quadratic cost. We plan to investigate this question, together with the setting of \cite{P16}, in future research.

\subsection{An empirical example}
Consider a smooth and exponentially concave function $\varphi$. It is $\alpha$-exponentially concave for all $0 < \alpha \leq 1$ and is concave (which corresponds to the case $\alpha \downarrow 0$). Thus both the additive and multiplicatively generated portfolios are well-defined. Unfortunately, while the $L^{(\alpha)}$-divergence is a natural interpolation,  there does not seem to be a canonical choice for the constant $C$ that connects the two basic cases.

 In this subsection we consider instead the parameterized family $\left\{ \eta^{(\alpha)} \right\}_{0 \leq \alpha \leq 1}$ where $\eta^{(\alpha)}$ is the trading strategy $(\alpha, \frac{1}{\alpha})$-generated by $\varphi$ (so when $\alpha = 0$ it is the additively generated portfolio), and compare their empirical performance. Note that $\eta^{(1)}$ is not the multiplicatively generated portfolio but is one which is more aggressive.

In this empirical example we let $n = 3$, so that we are able to visualize the path of the market in the simplex $\Delta_3$. We consider the (beginning) monthly stock prices of the US companies Ford, Walmart and IBM from January 1990 ($t = 0$) to  September 2017 ($t = 332$). We normalize the prices so that at $t = 0$ the market weight is at the barycenter $\left(\frac{1}{3}, \frac{1}{3}, \frac{1}{3}\right)$. The path of the market weight $\mu(t)$ in the simplex $\Delta_3$ is plotted in Figure \ref{fig:output} (left).

\begin{figure}[t!]
\centering
\includegraphics[scale=0.5]{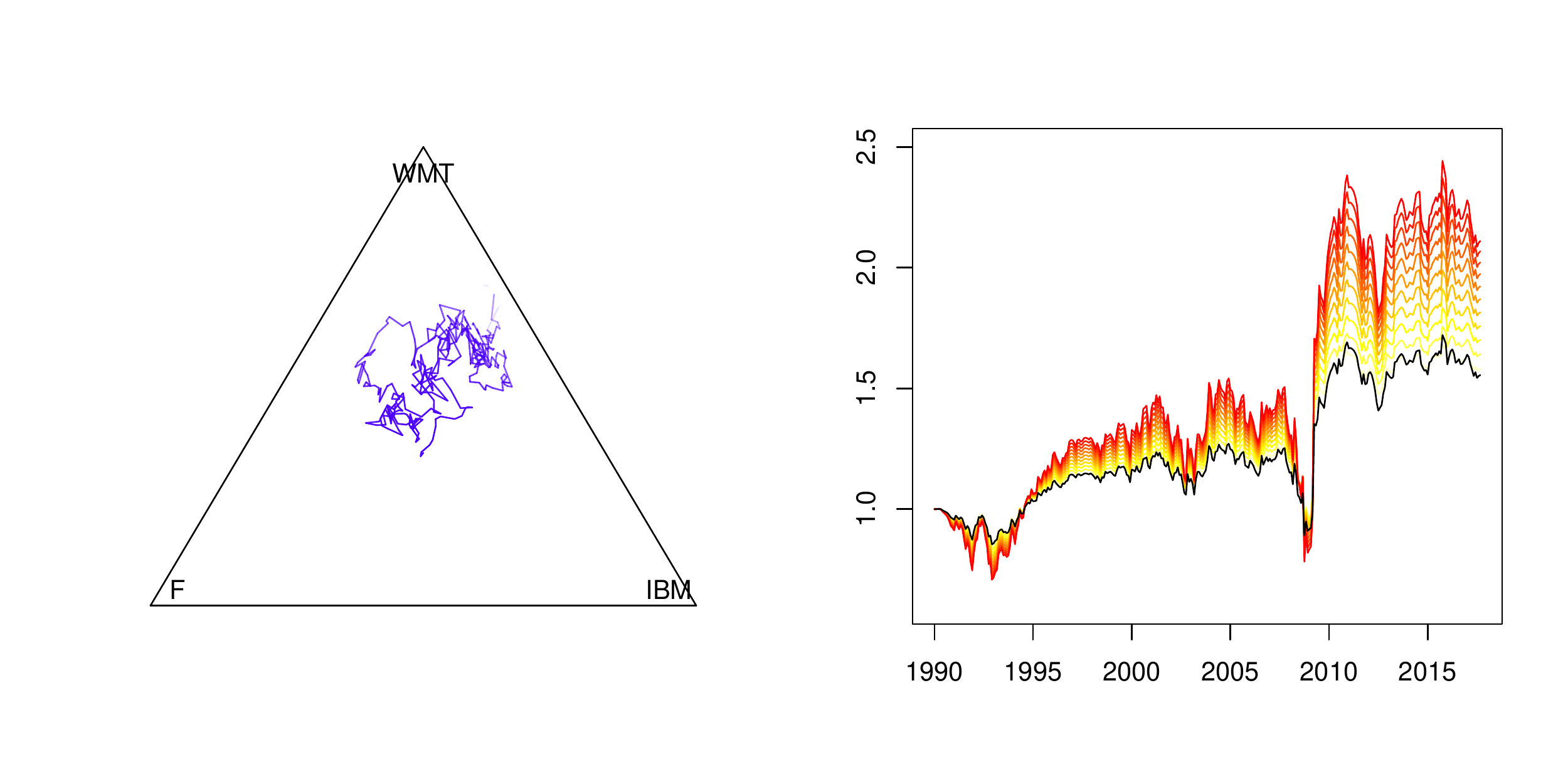}
\caption{Left: The market path $\{\mu(t)\}_{t = 0}^{332}$ in the simplex $\Delta_3$. The color is chosen according to the current value of $\varphi(\mu(t))$ (which is maximum at the barycenter, when $t = 0$). Right: Time series of the portfolio value $V_{\eta^{(\alpha)}}(t)$, from $\alpha = 0$ (yellow) to $\alpha = 1$ (red). The value of the equal-weighted portfolio is shown in black.}
\label{fig:output}
\end{figure}

The generating function chosen is the cross entropy
\begin{equation}
\varphi(p) = \sum_{i = 1}^n \frac{1}{3} \log p_i
\end{equation}
which generates multiplicatively the equal-weighted portfolio $\boldsymbol{\pi}(p) \equiv \overline{e} := \left(\frac{1}{3}, \frac{1}{3}, \frac{1}{3}\right)$. By \eqref{eqn:general.generation}, for each $\alpha \in [0, 1]$ the trading strategy is given by
\[
\eta^{(\alpha)}_i(t) = \left(1 + \alpha V_{\eta}(t)\right) \left(\frac{1}{3 \mu_i(t)} - 1\right) + V_{\eta}(t).
\]
In terms of portfolio weights, we have
\[
\pi^{(\alpha)}(t) = \frac{1 + \alpha V_{\eta}(t)}{V_{\eta}(t)} \overline{e} - \frac{1 + \alpha V_{\eta}(t) - V_{\eta}(t)}{V_{\eta}(t)} \mu(t).
\]
Thus the portfolio longs more and more the equal-weighted portfolio as $\alpha$ increases. We also set $V_{\eta}(0) = 1$. The corresponding $L^{(\alpha)}$-divergence is given by
\[
D^{(\alpha)}\left[q \mid p\right] = \frac{1}{\alpha} \log \left(1 + \alpha \sum_{i = 1}^n \frac{1}{n p_i} (q_i - p_i)\right) - \sum_{i = 1}^n \frac{1}{n} \log \frac{q_i}{p_i}.
\]

The values of the simulated portfolios are plotted in Figure \ref{fig:output} (right). At the end of the period the portfolio value is increasing in $\alpha$, and the additive portfolio ($\alpha = 0$) has the smallest value. It is interesting to note that the reverse is true at the beginning. Note that the values fluctuate widely in the period 2008--2009 corresponding to the financial crisis. For comparison, we also simulate the equal-weighted portfolio (i.e., $(\alpha, C) = (1, 0)$). Interestingly, the additive and multiplicative portfolios have similar behaviors here. In this period, shorting the market by using a positive value for $C$ gives significant advantage over both the additive and multiplicative portfolios.

Intuitively, as the parameter $\alpha$ decreases from $1$ to $0$, the effect of the divergence changes from compounding to addition, as can be seen from the decomposition \eqref{eqn:general.decomp}. If the divergence term grows roughly linearly in time (see \cite{FK05}), the additive portfolio will underperform those with $\alpha > 0$ in the long run, at least when $\varphi(\mu(t))$ is stable. Expressing this in another way, additively generated portfolios may be better over short horizons. Dynamic optimization over our extended functionally generated portfolios is an interesting problem.

\section*{Acknowledgement}
The author would like to thank Prof.~Ioannis Karatzas and Prof.~Johannes Ruf for helpful comments on the manuscript.

\bibliographystyle{abbrv}
\bibliography{referencesSPT}

\end{document}